\documentclass{llncs}

\usepackage{graphicx}
\usepackage{xcolor}
\usepackage{listings}
\usepackage{mips}
\usepackage{amsmath}
\usepackage{cancel}
\usepackage{tikz}
\usepackage{url}
\usetikzlibrary{shapes,backgrounds,arrows, fit, positioning}
\usepackage{algpseudocode}

\usepackage{geometry}
\newcommand{\crule}[3][black]{\textcolor{#1}{\rule{#2}{#3}}}

\newcommand{\erel}{\textsf{equiv}}
\newcommand{\merel}[1]{\textsf{equiv}_{#1}}
\newcommand{\pequiv}[2]{#1 \ \erel \ #2}
\newcommand{\pmequiv}[3]{#1 \ \merel{#3} \ #2}
\newcommand{\pnequiv}[2]{#1 \ \cancel{\erel} \ #2}
\newcommand{\pnmequiv}[3]{#1 \ \cancel{\merel{#3}} \ #2}

\definecolor{CommentGreen}{rgb}{0,.6,0}
\title{
Tamarin: Concolic Disequivalence for MIPS
\\
\large Technical Report}
\author{Abel Nieto}
\institute{
University of Waterloo\\
\email{anietoro@uwaterloo.ca}
}

\begin{document}

\maketitle

\begin{abstract}
Given two MIPS programs, when are they equivalent? At first glance, this is tricky to define, because of the unstructured nature of assembly code. We propose the use of alternating concolic execution to detect whether two programs are disequivalent. We have implemented our approach in a tool called Tamarin, which includes a MIPS emulator instrumented to record symbolic traces, as well as a concolic execution engine that integrates with the Z3 solver. We show that Tamarin is able to reason about program disequivalence in a number of scenarios, without any a-priori knowledge about the MIPS programs under consideration.
\end{abstract}

\section{Introduction}

We are staring at two opaque black boxes laying at our feet. Each box has a narrow slot through which we can place items in the box, but we cannot quite see what is inside. They look approximately like this:

\vspace{1mm}
\crule{1.5cm}{1.5cm}, \crule{1.5cm}{1.5cm}
\vspace{1mm}

We know each box contains an animal, but we do not know which specific animal is in each one. We would like to find out if both boxes contain the same species of animal. Our solution is simple: we take two carrots, and drop one in each box through the slots.

After a while, a chewing sound emerges from the boxes. We peer into them and, indeed, it looks like the carrots were successfully eaten. Triumphantly, we declare that the boxes contain the same species of animal. The truth is altogether different:

\vspace{1mm}
\fbox{\includegraphics[width=1.5cm]{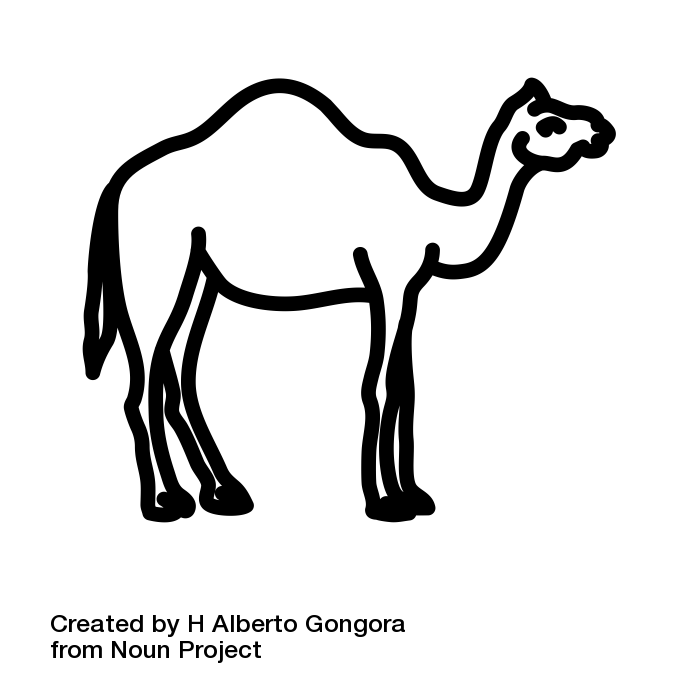}}, \fbox{\includegraphics[width=1.5cm]{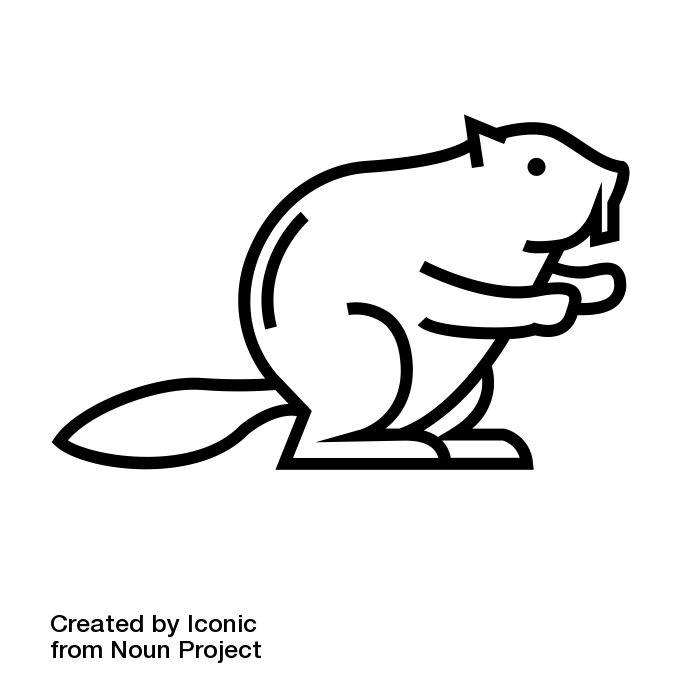}}
\vspace{1mm}

The boxes are assembly programs. The animals are the functions those programs compute. The carrot is unit testing. The task was to determine whether the programs were equivalent. And we failed at it. In this paper, we show a technique that is better than a carrot.

\section{Program Equivalence for MIPS}
\label{progequiv}

Let us set up the problem a bit more formally. Consider the set $P$ of MIPS-assembly programs that satisfy two restrictions: they take as inputs only the values of registers $\$1$ and $\$2$, and when they stop executing we define their output to be (exclusively) the value of $\$3$. Other side effects, such as printing values to the screen, or system calls, are disallowed.

We can now define a relation $\erel \subseteq P \times P$ (and its complement, $\cancel{\erel}$) of equivalent programs. Given $P_1, P_2 \in P$, we say that $\pequiv{P_1}{P_2}$ (read ``$P_1$ is equivalent to $P_2$'') if, for all inputs $\$1$ and $\$2$, one of the following holds:
\begin{itemize}
\item Both $P_1$ and $P_2$ fail during execution (for example, due to a divide-by-zero error).
\item $P_1$ and $P_2$ stop with the same output in $\$3$. 
\end{itemize} 

For example, the two programs in Figure \ref{equivprogs} are equivalent.
 
\begin{figure}
\begin{multicols}{2}
\begin{lstlisting}[escapeinside={(*}{*)}]
# P_1
add $3, $1, $2
\end{lstlisting}
\vfill\null
\columnbreak
\begin{lstlisting}
# P_2
add $4, $1, $1
lis $5
42
sw $4, 0, $5
add $3, $1, $2
\end{lstlisting}
\end{multicols}
\caption{$\pequiv{P_1}{P_2}$}
\label{equivprogs}
\end{figure}

Notice that $\pequiv{P_1}{P_2}$ even though $P_2$ modifies the contents of the memory and an additional register ($\$4$), because:
\begin{itemize}
\item Both $P_1$ and $P_2$ terminate without errors.
\item The value of $\$3$ will be the same when they do so.
\end{itemize}

Unfortunately, even though $\erel$ captures an already-simplified notion of equivalence\footnote{For example, $\erel$ has a very narrow notion of output that excludes side effects.}, a decision procedure for it does not exist, due to Rice's theorem.

To get decidability back, we define a new class of relations $\merel{S} \subseteq P \times P$ (whose complement is $\cancel{\merel{S}}$). We say that $\pmequiv{P_1}{P_2}{S}$ (read ``$P_1$ is $S$-equivalent to $P_2$'') if, for all inputs, one of the following holds:
\begin{itemize}
\item Either $P_1$ or $P_2$ does not stop within $S$ steps (we can think of each CPU cycle as one step).
\item Both $P_1$ and $P_2$ fail.
\item Both $P_1$ and $P_2$ stop with the same output. 
\end{itemize}

The $\merel{S}$ relation captures the notion that we cannot tell $P_1$ and $P_2$ apart by running them for at most $S$ steps. Figure \ref{sequiv} shows an example of two programs that are $S$-equivalent for $S=10$, but not equivalent. This is the case because $P_2$ loops while the counter is less than $42$, so with 10 steps in our ``budget'' we will have to stop $P_2$ before the loop is over and we can observe the different result.

\begin{figure}
\begin{multicols}{2}
\begin{lstlisting}[escapeinside={(*}{*)}]
# P_1
add $3, $1, $2
\end{lstlisting}
\vfill\null
\columnbreak
\begin{lstlisting}
# P_2
  add $4, $0, 1  # counter
  add $5, $0, 42 # upper bound
loop:
  slt $6, $4, $5 
  beq $6, $0, end
  add $4, $4, 1
  beq $0, $0, loop
end:
  add $3, $1, $1  
\end{lstlisting}
\end{multicols}
\caption{$\pmequiv{P_1}{P_2}{10}$, but $\pnequiv{P_1}{P_2}$}.
\label{sequiv}
\end{figure}

Given a fixed $S$, the $\merel{S}$ relation is decidable because there is a finite number of inputs to try, and for each input we only need to run the programs a finite number of steps.

We already saw that equivalence not always implies $S$-equivalence. However, the converse always holds. The following lemma shows that $\merel{S}$ over-approximates $\erel$.

\begin{lemma}
\label{overapprox}
$\forall S, P_1, P_2$, $\pequiv{P_1}{P_2} \implies \pmequiv{P_1}{P_2}{S}$.
\end{lemma}

\begin{proof}
Let $\pequiv{P_1}{P_2}$. Consider two arbitrary inputs $x$ and $y$. Then we have one of two cases:
\begin{itemize}
\item Either $P_1$ or $P_2$ (or both) do not stop within S steps when run on $x$ and $y$. This is the first case in the definition of $\merel{S}$.
\item Both $P_1$ and $P_2$ stop within S steps. Then because they are equivalent, we know that they either fail with an error, or both stop with the same output. These are the second and third cases in the definition of $\merel{S}$.
\end{itemize}

Therefore, we must have $\pmequiv{P_1}{P_2}{S}$.
\end{proof}

\begin{corollary}
\label{soundness}
$\pnmequiv{P_1}{P_2}{S} \implies \pnequiv{P_1}{P_2}$. 
\end{corollary}

\begin{proof}
This is just the contrapositive of Lemma \ref{overapprox}.
\end{proof}

Corollary \ref{soundness} can be used to argue the soundness (with respect to $\cancel{\erel}$) of any decision procedure that under-approximates $\cancel{\merel{S}}$. In the next section we will show one such under-approximation based on concolic execution.

\section{Concolic Disequivalence}
\label{ideasection}

We know from Corollary \ref{soundness} that any relation that under-approximates $\cancel{\merel{S}}$ is sound. Figure \ref{hierarchy}  shows why we want an under-approximation: efficiency. $\cancel{\erel}$ captures the class of programs that are disequivalent, but is undecidable. $\cancel{\merel{S}}$ is decidable, but likely cannot be computed efficiently. Therefore, we look for a subset of $\cancel{\merel{S}}$ (an under-approximation) that can be efficiently computed.

\begin{figure}
\def\firstcircle{(0,-1.8cm) circle (1.2cm)}
\def\secondcircle{(0,-1cm) circle (2cm)}
\def\thirdcircle{(0,0) circle (3cm)}
\begin{tikzpicture}
    \begin{scope}
        \draw[fill=green,opacity=0.1] \firstcircle;
        \draw[fill=yellow,opacity=0.1] \secondcircle;
        \draw[fill=red,opacity=0.1] \thirdcircle;
        \draw (0, -2) node  {??? (efficient)};
        \draw (0, 0) node {$\cancel{\merel{S}}$ (inefficient)};
        \draw (0, 2) node  {$\cancel{\erel}$ (undecidable)};
    \end{scope}
\end{tikzpicture}
\label{hierarchy}
\caption{Hierarchy of disequivalence relations}
\end{figure}
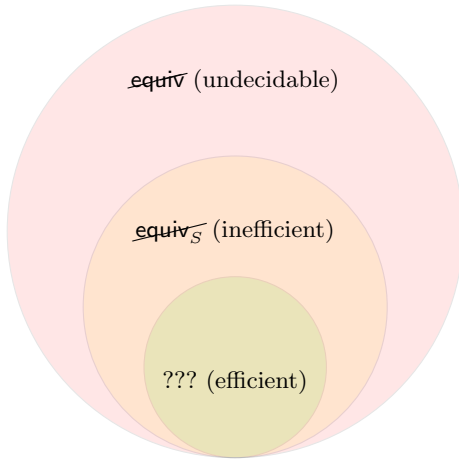

To fill the missing relation in Figure \ref{hierarchy} we propose concolic disequivalence. Abstractly, concolic disequivalence is a function $\textsf{compare}(P_1, P_2, S)$ that takes as inputs two MIPS programs and returns one of two answers:
\begin{itemize}
\item ``disequivalent'', in which case $\pnmequiv{P_1}{P_2}{S}$.
\item ``possibly equivalent'', meaning that $P_1$ and $P_2$ might or might not be $S$-equivalent.
\end{itemize}

Figure \ref{concolicalgo} shows pseudocode for \textsf{compare}. The algorithm alternately executes $P_1$ and $P_2$. At every step, one of the programs is labelled as the ``driver'' and the other one as the ``verifier''. The driver program is then concolically executed, yielding a set of inputs that exercise a new program path (of the driver). The inputs can then be fed to the verifier, and the results of both driver and verifier compared. If the results are different, then we know $P_1$ and $P_2$ are disequivalent. Otherwise, the driver becomes the verifier, and vice-versa. Eventually, we will traverse all explorable paths, at which point $P_1$ and $P_2$ can be declared possibly equivalent.

\begin{figure}
\begin{algorithmic}
\Function{compare}{$P_1$, $P_2$, $S$}
\State $b \gets true$
\While{either $P_1$ or $P_2$ has unexplored paths}
  \If{$b$} \Comment{Select driver and verifier}
    \State $D \gets P_1$
    \State $V \gets P_2$
  \Else
    \State $D \gets P_2$
    \State $V \gets P_1$
  \EndIf
  \If{$D$ has unexplored paths}
    \State $I \gets$ new inputs that exercise an unexplored path
    \State $R_1 \gets \textsf{run}(P_1, I, S)$
    \State $R_2 \gets \textsf{run}(P_2, I, S)$
    \If{both $P_1$ and $P_2$ stopped}
      \If{both $P_1$ and $P_2$ stopped with an error}
      \State \Comment{do nothing}
      \ElsIf{either $P_1$ or $P_2$ stopped with an error}
        \State \Return{``disequivalent''}
      \Else
        \If{$R_1 \neq R_2$}
          \State \Return{``disequivalent''}
        \EndIf
      \EndIf
    \EndIf
    \State mark the path discovered by $I$ as explored
    \State $b \gets \neg b$
  \EndIf
\EndWhile
\State \Return{``possibly equivalent''}
\EndFunction
\end{algorithmic}
\label{concolicalgo}
\caption{Concolic disequivalence algorithm}
\end{figure}

We now give an example of how \textsf{compare} operates. Consider the sample programs below:

\begin{multicols}{2}
\begin{lstlisting}[escapeinside={(*}{*)}]
# P_1
  bne $1, 42, end
  add $3, $3, $0
end:
  add $3, $1, $2
\end{lstlisting}
\vfill\null
\columnbreak
\begin{lstlisting}
# P_2
  add $3, $1, $2  
  bne $2, 100, end
  add $3, $3, $2
end:
\end{lstlisting}
\end{multicols}

Figure \ref{algoruns} summarizes the state of the algorithm as it compares $P_1$ and $P_2$. First, notice how the driver and verifier roles flip between $P_1$ and $P_2$ in consecutive runs. Every row indicates the input values, as well as the outputs $R_D$ and $R_V$ of the driver and verifier, respectively. At every run, we also record the path taken by the driver. Path conditions are negated to make sure we explore new paths in every iteration. In the fourth iteration, we can see that compare finds that the input pair $\$1 = 1, \$2 = 100$ leads to different outputs in the driver and verifier. At this point, $P_1$ and $P_2$ are declared as disequivalent. 

Notice that in order to uncover the different, it is necessary to concolically explore the paths in both $P_1$ and $P_2$, and not only of $P_1$. In Figure \ref{algoruns}, runs 1 and 3 explore both branches of the conditional jump in $P_1$, but they exercise the same path in $P_2$. Only after we also execute $P_2$ do we find a counterexample to equivalence.

\begin{figure}
\begin{tabular}{l | l | l | l | l | l | l | l}
\textbf{Run} & \textbf{Driver} & \textbf{Verifier} & $\mathbf{\$1}$ & $\mathbf{\$2}$ & \textbf{Path} & $\mathbf{R_D}$ & $\mathbf{R_V}$ \\
\hline 
1 & $P_1$ & $P_2$ & 1 & 1 & $\$1 \neq 42$ & 2 & 2 \\
2 & $P_2$ & $P_1$ & 1 & 1 & $\$2 \neq 100$ & 2 & 2\\
3 & $P_1$ & $P_2$ & 42 & 1 & $\$1 = 42$ & 2 & 2 \\
4 & $P_2$ & $P_1$ & 1 & 100 & $\$2 = 100$ & \textcolor{red}{201} & \textcolor{red}{2} 
\end{tabular}
\caption{A sample execution of \textsf{compare}}
\label{algoruns}
\end{figure}

\section{Tamarin}

Tamarin\footnote{Tamarins are small-sized monkeys from Central and South America. They are related to marmosets, which are also New World monkeys, and less-importantly give name to the black-box submission and testing server in use at the University of Waterloo as of Fall 2017 \cite{spacco2006marmoset}.} is a Scala implementation of the \textsf{compare} algorithm from Section \ref{ideasection}. We first give an overview of the major modules in Tamarin, shown in Figure \ref{tamarinstructure}, and then describe them in more detail in subsequent sections:

\begin{itemize}
\item \textsf{Concolic} is the entry point to Tamarin. It implements the top-level loop that visits unexplored paths, alternating between the two programs being compared. \textsf{Concolic} uses the other modules as helpers (in Figure \ref{tamarinstructure}, requests made by a module appear as solid lines, and responses to prior requests are shown with dashed lines).
\item \textsf{CPU} is a MIPS emulator that is instrumented to record symbolic traces containing path conditions, which can later be negated to explore new paths.
\item \textsf{Trace} consumes raw traces coming from \textsf{CPU} and transforms them in multiple ways so that they can be handed over to the \textsf{Z3} solver.
\item \textsf{Query} translates (modified) \textsf{CPU} traces into equivalent logical formulae. The formulae are then solved by the \textsf{Z3} solver, producing new inputs that, if fed to the program under test, will lead to traversing unexplored paths.
\item \textsf{Z3} is an SMT solver developed at Microsoft \cite{de2008z3}. We use it as black box for solving queries, over the theories of bitvectors and arrays, that result from program traces.
\end{itemize}

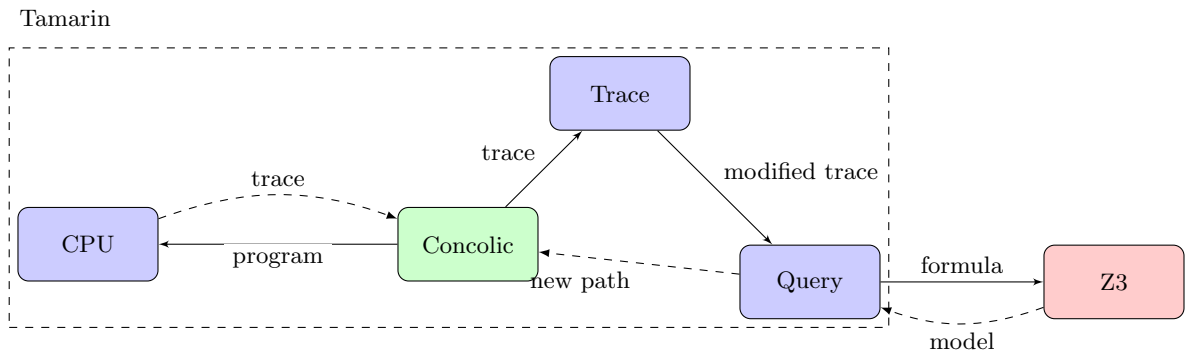
\begin{figure}
\tikzstyle{block} = [rectangle, draw, fill=blue!20, 
    text width=5em, text centered, rounded corners, minimum height=3em]
    \tikzstyle{mainblock} = [rectangle, draw, fill=green!20, 
    text width=5em, text centered, rounded corners, minimum height=3em]
\tikzstyle{line} = [draw, -latex']
    
\begin{tikzpicture}[node distance = 2cm, auto]
    \node [mainblock] (concolic) {Concolic};
    \node [block, left of=concolic, node distance=5cm] (cpu) {CPU};
    \node [block, above of=concolic, right of=concolic, node distance=2cm] (trace) {Trace};
    \node [block, below of=trace, right of=trace, node distance=2.5cm] (query) {Query};
    \node [block, fill=red!20, right of=query, node distance=4cm] (z3) {Z3};
    \node [draw=black, dashed, fit= (concolic) (cpu) (trace) (query)] {};
    \node at (-5.3, 3) {Tamarin};
    
    \draw [line] (concolic) -- (cpu) node [midway, fill=white] {program};
    \draw  (cpu)  edge [-latex, dashed, bend left=20] node[midway] {trace} (concolic);
    \draw [line] (concolic) -- (trace) node [midway]{trace};
    \draw [line] (trace) -- (query) node [midway]{modified trace};
    \draw [-latex, dashed] (query) -- (concolic) node [midway]{new path};
    \draw [line] (query) -- (z3) node [midway]{formula};
    \draw [-latex, dashed] (z3) edge [-latex, dashed, bend left=20] node[midway]{model} (query);
\end{tikzpicture}
\caption{Overview of Tamarin's modules}
\label{tamarinstructure}
\end{figure}

\subsection{Trace Collection}

The \textsf{CPU} module is in charge of running MIPS programs and collecting symbolic traces from executions. It is based on the MIPS emulator written by Ond\v{r}ej Lhot\'{a}k for CS241E at the University of Waterloo \cite{cs241e}.

The interface to the module consists of a single function:

\begin{lstlisting}[language=scala]
def run(prog: Seq[Word], r1: Word, r2: Word, fuel: Long): RunRes
\end{lstlisting}

The \textsf{run} function takes as input a program represented as as sequence of words, the values of registers $\$1$ and $\$2$ (the inputs to the program) and a \textsf{fuel} value (explained below).

The output is an algebraic data type \textsf{RunRes} that can take one of three forms:

\begin{lstlisting}[language=scala]
trait RunRes
case class Done(state: State, trace: Trace) extends RunRes
case class NotDone(trace: Trace) extends RunRes
case class Error(ex: RuntimeException) extends RunRes
\end{lstlisting}

\begin{itemize}
\item If the program executes without error, then \textsf{Done(state, trace)} is returned. \textsf{state} is the state of the CPU after execution, including the contents of memory (which are ignored) and of register $\$3$, the output register. \textsf{trace} is the symbolic trace captured during the program's execution, and is described below.
\item If the program ran for more than \textsf{fuel} CPU cycles without stopping, then the result is \textsf{NotDone(trace)}. Notice that even though the program did not stop we can still return a trace recording the execution right until the moment we stopped it. The \textsf{fuel} argument to \textsf{run} plays the same role as the $S$ argument in Figure \ref{concolicalgo}. 
\item Finally, if there was an error during program execution (for example, an attempted division-by-zero), then we return \textsf{Error}.
\end{itemize}

Notice that due to \textsf{fuel} parameter and error boxing, the augmented emulator in Tamarin, unlike a vanilla MIPS emulator, is ``hardened'' in the sense that it can execute MIPS programs that do not stop (the emulator itself will stop) or throw errors (will be catched at the top level by the emulator). 

While the emulator is executing a program, it also records a symbolic trace of (most of) the executed instructions. We support a subset of the MIPS instruction set, containing 18 instruction types \cite{mipsinstructions}. Notably, unlike in full MIPS, there are no system calls in our supported subset.

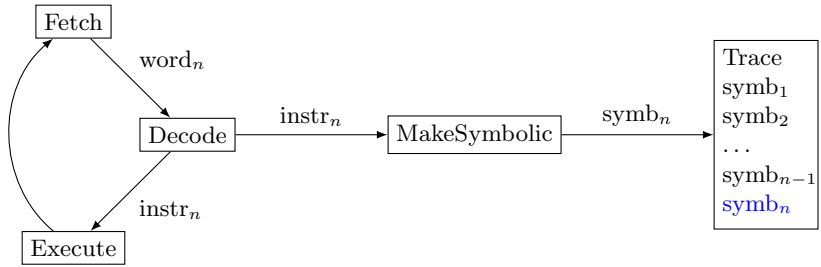
\begin{figure}
\begin{tikzpicture}[node distance = 1.5cm, auto]
\node [draw] (fetch) {Fetch};
\node [draw, below of=fetch, right of=fetch] (decode) {Decode};
\node [draw, below of=decode, left of=decode] (execute) {Execute};
\node [draw, right=2cm of decode] (symb) {MakeSymbolic};
\node [draw, right=2cm of symb, minimum height=2.5cm, text width=1.15cm] (trace) {};
\node [below right, text width=1cm] at (trace.north west) {Trace symb$_1$ symb$_2$ $\ldots$ symb$_{n - 1}$ \textcolor{blue}{symb$_n$}};

\draw [-latex] (fetch) -- (decode) node [midway] {word$_n$};
\draw [-latex] (decode) -- (execute) node [midway] {instr$_n$};
\draw [-latex, bend left=50] (execute) edge (fetch);
\draw [-latex] (decode) -- (symb) node [midway] {instr$_n$};
\draw [-latex] (symb) -- (trace) node [midway] {symb$_n$};
\end{tikzpicture}
\caption{Instrumented CPU that records symbolic traces}
\label{cpucycle}
\end{figure}

The fetch-decode-execute cycle in our emulator is modified to record a symbolic representation of each instruction as it is executed (Figure \ref{cpucycle}). The symbolic instructions are then stored in a trace. There are two types of instructions in a trace: assignments and path conditions.
\begin{itemize}
\item Assignments are instructions that mutate the CPU state, but do not affect the control flow. The symbolic form of most assignments is very similar to the concrete instruction that is executed. For example, the symbolic representation of the instruction \lstinline{add $3, $1, $2} is $r_3 \gets r_1 + r_2$. 

\item Path conditions are instructions that modify the control flow of the program. This set potentially contains both conditional and unconditional jumps, but we track only conditional ones. Specifically, we record the branch taken at each conditional jump, and symbolically record it as an (in)equality. For example, if the branch \lstinline{beq $1, $2, foolabel} is not taken, then we will add $r_1 \neq r_2$ to our trace.
\end{itemize}

One additional point of note: the program counter (\textsf{pc}) is not symbolically tracked. This means Tamarin cannot reason about unconditional jumps, so only a single path is explored for the program in Figure \ref{pcnottrack}. There, Tamarin will execute $P$ with the initial inputs (a hard-coded constant) and error out, because the jump will lead to an invalid instruction. Tamarin will miss the fact that $\$1$ and $\$2$ can be given values such that we could jump to either of the labels, leading to a successful  execution. 

\begin{figure}
\begin{lstlisting}
# P
  add $4, $1, $2
  jr $4
good:
  add $3, $1, $2
  jr $31 # $31 contains the termination pc
bad:
  add $3, $1, $1
  jr $31    
\end{lstlisting} 
\caption{Not tracking the \textsf{pc} leads to under-approximations}
\label{pcnottrack}
\end{figure}

Since the \textsf{pc} is not tracked, whenever it is used in another instruction we need to substitute it by its concrete value. This technique is called concretization \cite{david2016specification}. For example, the load immediate and skip instruction has the following semantics:

\lstinline{lis $d} \hspace{2cm} $d \gets Mem[pc]; pc \gets pc + 4$

Instead of the (precise) interpretation above, we use concretization to generate the (under-approximating) symbolic instruction $d \gets v$, where $v$ is the value at the address following the \lstinline{lis} instruction.

\subsection{Transformations}
\label{transsection}

The \textsf{Trace} module also exposes just a single function:

\begin{lstlisting}[language=scala]
def transform(trace: Trace, depth: Int): Trace
\end{lstlisting}

\textsf{transform} takes as input a trace that was produced by \textsf{CPU} and a \textsf{depth} parameter (explained below), and modifies the input trace so that it can be later converted into a logical formula. The \textsf{depth} parameter limits the number of path conditions in the outputted trace: this is so we can bound the depth of our DFS as we explore new paths (see Section \ref{redux} for more details). In effect, we can think of traces as the intermediate representation (IR) for queries to the SMT solver. The \textsf{Trace} module can be then thought of as lowering the semantic complexity of traces.

Internally, \textsf{Trace} consists of multiple phases, each of type \lstinline{Trace => Trace}, which are sequentially applied to the input. They are described below.

\subsubsection{Desugaring}

 The desugaring phase removes ``complicated'' instructions by replacing them with simpler ones. Figure \ref{desugar} summarizes some of the transformations.
 
\begin{figure}
\begin{tabular}{l | l}
\textbf{Input} & \textbf{Output} \\
\hline
$Jalr(concretePC)$ & $Add(\$31, concretePC)$ \\
$Mult(s, t)$ & $Mult64(tmp, s, t); Low32(lo, tmp); High32(hi, tmp)$ \\
$Div(s, t)$ & $Quot(lo, s, t); Rem(hi, s, t)$ \\
$Mflo(d)$ & $Add(d, lo, 0)$
\end{tabular}
\caption{Desugared instructions}
\label{desugar}
\end{figure}

A couple of points of interest:
\begin{itemize}
\item $\$31$ holds by convention the return address of the current procedure. Since the \textsf{pc} is not tracked, its concretized value is used.

\item Notice how some instructions (e.g. $Mult$) are desugared into multiple instructions: this is because they in fact have more than one side effect in the CPU state.

\item While desugaring some instructions, we introduce ``helper'' instructions like $Mult64$ and $Quot$ that are not part of the MIPS instruction set. These instructions help communicate the intended semantics to the SMT solver. For example, later on we will see that the contents of registers are represented with 32-bit vectors, but multiplication of two 32-bit integers can take up to 64 bits of storage, so we need a way to let the solver know that it should upcast the product of two registers (so we can later decompose it into the low and high 32-bits): the $Mult64$ instruction accomplishes this.
\end{itemize}

\subsubsection{Simplification}

The simplification phase removes trivial path conditions, like \lstinline{beq 0, 0, foolabel}, which might be common (the specific \lstinline{beq} example is a common idiom for jumping to a loop header) but do not need to be considered by the SMT solver.

\subsubsection{Trimming}

The trimming phase just shortens the trace so that it contains at most \textsf{depth} (an argument to \textsf{Trace}) path conditions. This is so we can do a bounded DFS (again, see Section \ref{redux}) to limit the search space.

\subsubsection{SSA Conversion}
\label{ssaconv}

The SSA phase converts traces into static single assignment (SSA) form \cite{cytron1991efficiently}. SSA form is crucial to be able to transform a trace into the corresponding logical formula. Consider the example in Figure \ref{ssa}, which shows a small trace together with incorrect and correct translations.

The incorrect translation represents assignment through equality. This is intuitive, but misguided, because assignment mutates state, but equality does not. For example, suppose we append to our trace the path condition $\$3 = 9 \land \$1 = 3$. This gets translated to the equality $z = 9 \land x = 3$. But the equalities in the incorrect translation imply that $y = 0$ (because $x = x - y$) and so $z = x = y$. The path condition is then declared as unsatisfiable, but this is clearly wrong, because $\$1 = 6 \land \$2 = 3$ is a solution. 

What went wrong, and how do we fix it? As we pointed, assignment mutates the value of registers, so we need to somehow encode that the ``$\$1$'' in the last $Add$ instruction is ``not the same'' as the one in the first $Add$. This behaviour is precisely captured by SSA form, where each variable is assigned to exactly once. This means that two reads of the same register that are separated by a write to the same register will be encoded as different variables. The correct translation in Figure \ref{ssa} uses subscripts to refer to the ``versions'' of each variable after writes. If we add the same path condition as before, we can see that it gets translated as $z_1 = 9 \land x_2 = 3$, which is now satisfiable because $x_2 = x_1 - y_1$ does not imply that $x_1 = y_1$.

\begin{figure}
\begin{tabular}{l | l | l}
\textbf{Trace} & \textbf{Incorrect} & \textbf{Correct} \\
\hline
$Add(\$3, \$1, \$2)$ & $z = x + y$ & $z_1 = x_1 + y_1$ \\
$Sub(\$1, \$1, \$2)$ & $x = x - y$ & $x_2 = x_1 - y_1$ \\
$Add(\$2, \$1, \$2)$ & $y = x + y$ & $y_2 = x_2 + y_1$\\
\end{tabular}
\caption{A trace together with an incorrect translation to a logical formula and a correct translation using SSA form}
\label{ssa}
\end{figure}

The algorithms for efficiently converting code into SSA form are sophisticated \cite{cytron1991efficiently}, requiring the calculation of so-called dominators and $\phi$-functions. However, we are only interested in converting traces, which do not contain any jumps. In this case, a simple linear-time pass over the trace suffices for converting to SSA. The algorithm is shown in Figure \ref{ssaalgo}.

\begin{figure}
\begin{algorithmic}
\Function{ssaConvert}{$trace$}
  \State $trace' \gets \emptyset$
  \State $\forall x. last(x) \gets x_1$  
  \ForAll{$instr \in trace$}
    \If{$instr = Add(d, s, t)$}
      \State $s' \gets last(s)$
      \State $t' \gets last(t)$
      \State $d_i \gets last(d)$
      \State $last(d) \gets d_{i + 1}$
      \State $instr' \gets Add(d_{i + 1}, s', t')$
      \State \textsf{append}$(trace', instr')$
    \Else $\text{ } \ldots$ \Comment{Other cases}
    \EndIf
  \EndFor
  \State \Return $trace'$
\EndFunction
\end{algorithmic}
\caption{SSA conversion of traces}
\label{ssaalgo}
\end{figure}

\subsection{Queries}

The \textsf{Query} module is responsible for interfacing with the \textsf{Z3} solver to determine the satisfiability of path constraints, leading to unexplored program paths.

The API for the module is the function

\begin{lstlisting}[language=scala]
def solve(trace: Trace): Option[Soln]
\end{lstlisting}

The input to \textsf{solve} is a trace that has been processed by \textsf{Trace}, likely containing a recently-negated path condition (negated by the \textsf{Concolic} module from Section \ref{redux}). \textsf{solve}'s task is to convert the trace into an equivalent logical formula, hand it to \textsf{Z3}, and then return the solution found by the solver.

A \textsf{Soln} is just a \lstinline{Seq[RegVal]} (\textsf{solve} returns an \textsf{Option[Soln] in case the formula is unsatisfiable}), and a \textsf{RegVal} is an algebraic datatype of the form

\begin{lstlisting}[language=scala]
case object Unbound extends RegVal
case class Fixed(v: Long) extends RegVal
\end{lstlisting}

A \textsf{RegVal} represents the value of a register in the model returned by the solver. \textsf{Unbound} represents the case when the solver does not constrain a register in its solution. By contrast, \textsf{Fixed} is used when the solver requires a specific value.

\subsubsection{Translation}

Figure \ref{smtlib} shows an example of how a program trace can be converted into a logical formula understood by the solver. In this case, the program has a single conditional branch, leading to one path condition. Presumably, we have already explored (perhaps from the initial run with random inputs) what happens when the branch is taken (that is $\$4 = \$5$ and so we skip the last addition). Now we are interested in exploring the case where the branch is not taken: for that, we need the path condition $\$4 \neq \$5$ at the end of the trace.

\begin{figure}
\begin{multicols}{3}
program:
\begin{lstlisting}
  add $3, $1, $2
  slt $4, $1, $2
  lis $5
  1
  beq $4, $5, skip
  add $6, $3, $1
skip:
\end{lstlisting}

\vfill\null
\columnbreak

trace:
\begin{lstlisting}
  add $3, $1, $2
  slt $4, $1, $2
  add $5, $1, $0
  $4 != $5
\end{lstlisting}

\vfill\null
\columnbreak

SMT-LIB:
\begin{lstlisting}[language=Lisp]
(declare-const r1 (_ BitVec 32))
(declare-const r2 (_ BitVec 32))
(declare-const r3 (_ BitVec 32))
(declare-const r4 (_ BitVec 32))
(declare-const r5 (_ BitVec 32))

(assert (= r0 (_ bv0 32)))
(assert (= r3 (bvadd r1 r2)))
(assert
  (= r4 (ite (bvslt r1 r2)
    (_ bv1 32)
    (_ bv0 32))))
(assert (= r5 (_ bv1 32)))
(assert (not (= r4 r5)))

(check-sat)
(get-model)
\end{lstlisting}
\end{multicols}
\caption{A sample program, a trace in it, and the trace's representation in SMT-LIB notation}
\label{smtlib}
\end{figure}

The third column in Figure \ref{smtlib} shows the logical formula that is produced by the \textsf{Query} module. The formula is satisfiable iff the path condition $\$4 \neq \$5$ can evaluate to true for certain values of $\$1$ and $\$2$. In Figure \ref{smtlib}, we have used the SMT-LIB \cite{barrett2010smt} representation of the query, which is a textual representation with a Lisp-like syntax. Tamarin uses Z3's Java API, for efficiency.

Every register in the trace has a corresponding ``constant'' in the formula ($r_1 \ldots r_6$). The constants have type ``32-bit bitvector'', to encode that registers are 32-bit integers. Consequently, arithmetic operations on registers are encoded as operations on bitvectors (e.g. \textsf{bvadd, bvslt}).

Some instructions like \lstinline{slt} have domain-specific semantics, so their translations are slightly more complicated: \lstinline{slt} specifically uses SMT-LIB's \textsf{ite} (if-then-else) construct.

There is no mutation in the formula, but that is ok because of our conversion to SSA from Section \ref{ssaconv}. Instead of mutation, relationships between constants are encoded using assertions. In general, we have one assertion per instruction in the trace.

Once all constant declarations and assertions have been specified, we can query Z3 to check the satisfiability of the formula via \textsf{(check-sat)}. If the formula is satisfiable, we can also request the satisfying model that Z3 found, with \textsf{get-model}. In this case, Z3 returns 

\begin{lstlisting}[language=Lisp]
(model (define-fun r1 () (_ BitVec 32) #x80000000)
       (define-fun r2 () (_ BitVec 32) #x80000000)
           ...)
\end{lstlisting}

Which indeed is a solution (albeit not one a human would have picked), because it makes \lstinline{slt $4, $1, $2} assign $0$ to $\$4$, which skips the branch.

\subsubsection{Memory Representation} 

In keeping with registers being represented as 32-bit bitvectors, Tamarin represents memory itself as an array of bitvectors. Below we show the translations of the \lstinline{sw} and \lstinline{lw} instructions:

\noindent \lstinline{lw $t, i, $s} \hspace{2cm} \lstinline{(assert (= r_t (select mem_last (bvadd i r_s))))} \\
\lstinline{sw $t, i, $s} \hspace{2cm} \lstinline{(assert (= mem_k+1 (store mem_k (bvadd i r_s) r_t))}

Notice how memory is also immutable: every \textsf{store} (\lstinline{sw}) operation returns a new memory constant (an array), and every \textsf{select} (\lstinline{lw}) uses the latest memory constant .

\subsection{Concolic Execution Redux}
\label{redux}

As mentioned before, the \textsf{Concolic} module implements the concolic execution engine and orchestrates the interactions between all modules.
We already presented the general structure in Figure \ref{concolicalgo}, but below we expand on some additional points.

\subsubsection{Alternation}

The main difference between Tamarin and a traditional concolic execution engine is that Tamarin needs to maintain state for the two programs it is comparing, and opposed to just one. Therefore, we represent the engine state with a tuple $(trace_1, trace_2, turn)$ of the last-executed trace for each program, together with a pointer to the last-run program. At every iteration, we use $turn$ to choose the program that is next in line for execution, and use the corresponding $trace_i$ to identify the last path-condition (and the corresponding inputs) that has not been previously negated. The programs are run on the new inputs, their traces recorded, and a new state $(trace_1', trace_2', turn')$ generated.

\subsubsection{Bounded DFS}

As in the original concolic testing tool, CUTE \cite{sen2005cute}, Tamarin explores program paths using a bounded depth-first search. To guarantee termination, Tamarin trims CPU traces so that they contain at most \textsf{depth} path conditions (later path conditions are simply ignored). The current depth is set at 50, but is configurable. 

\subsubsection{Comparing Results}

When are the results of the two programs consider contradictory? Only when both programs terminate with different values in $\$3$, or when one program fails while the other succeeds. In particular, non-termination is treated conservatively and no inferences are derived from it (this could be modified soundly if one of the two programs is given a ``priviledged'' status by considering it the specification). The full rules are those in Section \ref{progequiv}.

\subsubsection{Soundness and Completeness}

As we saw in Lemma \ref{soundness}, Tamarin is sound, but not complete. There are multiple sources of incompleteness: not all paths are explored (because of the \textsf{depth} parameter), termination is ensured via the \textsf{fuel}, and part of the program state is simply not tracked at all, leading to under-approximation.

\subsubsection{Efficiency}

We have not done an formal complexity analysis of the algorithm implemented by Tamarin. In any case, as with other concolic execution tools, there is a risk of combinatorial path explosion, which can be traded away in exchange for a loss in precision via the \textsf{depth} parameter. A similar tradeoff exists with the \textsf{fuel} parameter. The theories used by Tamarin via Z3 are quanitifier-free bitvectors and arrays, including multiplication, which is potentially problematic. If further testing revealed an impact to performance, we would look into concretization strategies to handle the multiplication case.

\section{Evaluation}
\label{evalsect}

To evaluate the correctness of Tamarin, we hand-wrote a set of assembly programs (program pairs, to be more precise) that exercise different functionality in the tool. The test programs are summarized in Figure \ref{testtable}.

\begin{figure}
\begin{tabular}{l | l}
\textbf{Program} & \textbf{Tested Property} \\
\hline
\textsf{stack} & Can reason about pushs and pops to the stack \\
\textsf{fun} & Finds counterexamples  across function calls \\
\textsf{while} & Able to reason about loops that iterate less than \textsf{depth} times \\
\textsf{infinite} & Treats infinite loops conservatively \\
\textsf{div0} & Can handle CPU exceptions \\
\textsf{nested} & Reasons about nested conditional statements
\end{tabular}
\caption{Test programs for Tamarin}
\label{testtable}
\end{figure}

The tests we have to date focus on correctness (and, indirectly, measure efficiency), but they are all small and simple programs. To evaluate whether Tamarin is of practical use, we plan to collect student-written compilers for an undergraduate course at the University of Waterloo. We can then feed these submitted compilers sample programs written in a Scala-like language. We can compare the MIPS code generated by the student's compilers to that of the reference implementation, using Tamarin. Finally, we will compare Tamarin's results with the current black-box testing approach used in the course, and investigate whether Tamarin provides better accuracy.

\section{Related Work}
We now survey some of the existing literature on program equivalence that is relevant to the project. In general, many tools have been built for checking program equivalence, differing on the language they target (high level vs assembly), their level of generality (one-off tools vs frameworks or even intermediate languages that serve as the backend for multiple tools), the language features they support (loops vs loop unrolling) and their soundness and completeness guarantees (concolic testing based tools vs complete exploration of a bounded search space). However, we were not able to find a tool that both does program equivalence and works with raw assembly programs.

The two ecosystems of tools most relevant to our work are KLEE and Boogie.

\subsubsection{KLEE}

\cite{cadar2008klee}  introduced KLEE, which is now a popular backend for program verification tasks. It implements a symbolic VM for LLVM bitcode. However, unlike Tamarin, the input to KLEE is at the C level. \cite{ramos2011practical} build UC-KLEE on top of KLEE. Given two C functions, UC-KLEE checks whether they are equivalent (up to a fixed input size of 8 bytes). Their notion of equivalence considers not only function return values, but also memory locations reachable from them. As far as we can tell, UC-KLEE was never released publicly.

\subsubsection{Boogie}

\cite{barnett2005boogie} present Boogie, an intermediate verification language that can be targeted as an IR for program verification tasks. \cite{lahiri2012symdiff} built SymDiff on top of Boogie to test for program equivalence. However, even though Boogie can encode MIPS, SymDiff requires that procedures in the two compared programs be given in pairs via a pre-established correspondence. This is so that Boogie can reason about one procedure at a time, making conservative assumptions about the state of the other procedures. This makes it unclear whether SymDiff could be used to verify equivalence of arbitrary MIPS programs whose structures we do not know a-priori. \cite{hawblitzel2013will} verify compiler correctness by checking program equivalence of emitted assembly code. They use x86, Boogie, and SymDiff.

\subsubsection{Other related work}

\cite{person2011directed} combine static analysis and symbolic execution for program equivalence. \cite{egele2014blanket} use \emph{Blanket Execution} to identify portions of program binaries that are ``'similar', based on randomized testing and having similar side-effects. \cite{partush2013abstract} use abstract interpretation to prove program equivalence for numerical programs. \cite{yang2014property} show how verify that manually-annotated program properties continue to hold as a program evolves. \cite{necula2000translation} verifies equivalence of IRs before and after compiler passes. \cite{francesco2014grease} present GreASE, which focuses on non-equivalence checking. \cite{sharma2013data} are able to check equivalence of x86 loops, unlike most of the other tools, which need to unroll loops.

\section{Future Work}

Tamarin is still very much a work-in-progress. Below we list some of the outstanding work.

\subsubsection{Data Structure Generation} The CUTE tool \cite{sen2005cute} is able to generate not only numeric inputs, but also C-like \textsf{structs} using a somewhat inductive approach: to generate a \textsf{struct}, first generate \textsf{null}, then generate an instance of the \textsf{struct} where all fields are \textsf{null}, and then iterate the construction. This is harder to do in Tamarin, because it is working at the assembly level. Unlike in C, where we have \textsf{struct} definitions to guide the ``shapes'' we need to generate, Tamarin has no type-level information to guide the creation of data structures. One option is to specify such shape with metadata that is given to Tamarin as input together with the two programs.

\subsubsection{Restarts} Because we do not precisely track the semantics of the executed MIPS program (for example, we ignore the \textsf{pc}), it is possible that we solve a set of path constraints, but the resulting input does not lead the MIPS program along the expected path. In these cases, Tamarin currently stops evaluating the ``misdirected'' program (we will however continue to concolically execute the second program). Another approach, also used in CUTE, is to restart the engine with random inputs. We do not currently know how common restarts are, but they might be worth implementing for more complicated programs.

\subsubsection{Experiments} As mentioned in Section \ref{evalsect}, we would like to run Tamarin with compiler-generated MIPS programs, to test scalability and real-world use. This would be easier to do if we could generate data structures, since many programs take more two integers as input.

\section{Conclusions}

Program equivalence is a hard problem. Equivalence of assembly language programs is arguably even harder, because it needs to handle the unstructured nature of assembly without type or grammar-level information.

Surprisingly, there does not seem to be a program equivalence tool that works at the assembly level (for arbitrary programs). Tamarin is one step in this direction. 

\bibliographystyle{apalike}
\bibliography{refs}

\end{document}